\documentclass[12pt]{iopart}
\usepackage{iopams,mathptm,times,graphicx}
\makeatletter
\newcounter{theorem}
\@addtoreset{theorem}{section}
\renewcommand\thetheorem{\arabic{section}.\arabic{theorem}}

\newenvironment{lemma}{\par\medskip\noindent\begingroup{\bf Lemma
             \stepcounter{theorem}\thetheorem.}\ \itshape
             \def\@currentlabel{\thetheorem}}{\endgroup\par\medskip}
\newenvironment{theorem}{\par\medskip\noindent\begingroup{\bf Theorem
             \stepcounter{theorem}\thetheorem.}\ \itshape
             \def\@currentlabel{\thetheorem}}{\endgroup\par\medskip}

\newenvironment{remark}{\par\medskip\noindent\begingroup{\bf Remark
             \stepcounter{theorem}\thetheorem.}\
             \def\@currentlabel{\thetheorem}}{\endgroup\par\medskip}

\newenvironment{proof}{\par\noindent{\bf Proof.} }{\proofbox\par\medskip}

\def\proofbox{\hfill{\ensuremath\Box}}

\newdimen\LENB \newdimen\LENW \newdimen\THI
\newdimen\LENWH \newdimen\LENTOT \newcount\N
\def\vbrknlnele#1#2#3{
  \LENB=#1pt \LENW=#2pt \THI=#3pt
  \LENWH=\LENW \divide\LENWH by 2
  \LENTOT=\LENB \advance\LENTOT by \LENW
  \vbox to \LENTOT{
    \vbox to \LENWH{}
    \nointerlineskip
    \vbox to \LENB{\hbox to \THI{\vrule width \THI height \LENB}}
    \nointerlineskip
    \vbox to \LENWH{}
  }}

\def\vbrknln#1{
  \N=#1
  \vcenter{
    \vbox{
      \loop\ifnum\N>0
        \vbox to 4pt{\vbrknlnele{2}{2}{0.1}}
        \nointerlineskip
        \advance\N by -1
      \repeat
  }}}

\def\hbrknlnele#1#2#3{
  \LENB=#1pt \LENW=#2pt \THI=#3pt
  \LENTOT=\LENB \advance\LENTOT by \LENW
  \vcenter{
    \vbox to \THI{
      \hbox to \LENTOT{
        \hfil
        \vrule width \LENB height \THI
        \hfil}
  }}}

\makeatletter
\usepackage{ulem}

\def\journal#1&#2,{\begingroup \let\journal=\dummyjournal
               \it #1\unskip~\bf\ignorespaces #2\rm,\endgroup}
\def\dummyjournal{\errmessage{Reference foul up: nested \journal macros}}

\eqnobysec

\def\eqref#1{(\ref{#1})}

\begin{document}
\title{Multi-soliton solution to the two-component Hunter-Saxton equation}
\author{Bao-Feng Feng$^{1}\footnote{e-mail: feng@utpa.edu}$, Senyue Lou$^2$
and Ruoxia Yao$^3$}
\address{$^1$~Department of Mathematics,
The University of Texas-Rio Grande Valley,
Edinburg, TX 78541}
\address{$^2$~Faculty of Science, Ningbo University, Ningbo 315211, China}

\address{$^3$~School of Computer Science,
Shaanxi  Normal University, Xi'an 710119, China}
\date{\today}
\begin{abstract}
In this paper, we study the bilinear form and the general $N$-soliton solution for a two-component Hunter-Saxton (2-HS) equation, which is the short wave limit of a two-component Camassa-Holm equation. By defining a hodograph transformation based on a conservation law and appropriate dependent variable transformations, we propose a set of bilinear equations which yields the 2-HS equation. Furthermore, we construct the $N$-soliton solution to the 2-HS equation based on the tau functions of an extended two-dimensional Toda-lattice hierarchy through reductions. One- and two-soliton solutions are calculated and analyzed.
\par
\kern\bigskipamount\noindent
\today
\end{abstract}


\section{Introduction}
The Hunter-Saxton (HS) equation
\begin{equation}
u_{txx}-2\kappa u_x+2u_xu_{xx}+uu_{xxx} =0,\label{HS}
\end{equation}
was derived as a model for propagation of orientation waves in a massive nematic
liquid crystal director field \cite{HS}. The Lax pair, the bi-Hamiltonian structure and local and global weak solutions were
discussed by Hunter and Zheng \cite{HunterZheng}. It belongs to the member of the hierarchy
of the Harry-Dym equation \cite{Cewen}, and in a series of papers by Alber and his collaborators, the link between billiard solutions and soliton solutions of HS equation were made clear \cite{Alber95,Alber99, Alber01}. The HS equation can be regarded as a short wave limit of the well-known Camassa-Holm equation
\cite{CH_Original,CHHyman}
\begin{equation}
m_t + u m_x+2m u_x=0, \quad m= \kappa+u-u_{xx}\,. \label{CH}
\end{equation}

In the present paper, we are concerned with the two-component Hunter-Saxton (2-HS) equation
\begin{equation}
u_{txx}-2\kappa u_x+2u_xu_{xx}+uu_{xxx} =\sigma \rho
\rho_x,\label{2CH-sweqa}
\end{equation}
\begin{equation}
\rho_t+(\rho u)_x=0,\label{2CH-sweqb}
\end{equation}
where $\kappa$ and $\sigma$ are positive constants. The 2-HS equation has attracted much attention in the past and it has been studied extensively
by many authors  \cite{Wunsch09,Lenells09,Wunsch10,Wunsch11,WunschWu,MoonLiu2012,Moon2013,Kohlmann,ChunhaiLi}. In a series papers by Wunsch, he studies the local and global weak solutions for the periodic 2-HS equation. The single solitary wave solution was studied in \cite{WunschWu,Moon2013,ChunhaiLi}. The local well-posedness and wave-breaking was studied by Moon and Liu \cite{MoonLiu2012}. The geometric property of the 2-HS equation was investigated in \cite{Kohlmann}.

The two-component Hunter Saxton system is a particular case of the Gurevich Zybin system  pertaining to nonlinear one-dimensional dynamics of dark matter as well as nonlinear ion-acoustic waves (cf.\cite{Pavlov} and the references therein). It was also derived
as the $N = 2$ supersymmetric extension of the CH equation \cite{Lenells09}. It is sometimes called the generalized two-component Hunter-Saxton system if $\kappa \ne 0$. We simply call it the two-component Hunter-Saxton equation for either $\kappa =0$ or $\kappa \ne 0$ hereafter.
The 2-HS equation can be viewed as the short wave limit of the two-component Camassa-Holm (2-CH) equation originating in the Green-Naghdi equations which approximate the governing equations for water waves \cite{YoujinLMP,AGZJPA,ConstantinIvanov}
\begin{equation}
m_t + u m_x+2m u_x-\sigma \rho \rho_x=0,  \label{2CHa}
\end{equation}
\begin{equation}
\rho_t+(\rho u)_x=0,  \label{2CHb}
\end{equation}
\begin{equation}
m= \kappa+u-u_{xx},  \label{2CHc}
\end{equation}
by using the scaling $(t,x) \to (\epsilon t, \epsilon^{-1}x)$. To be specific, let
\begin{equation}
\tau = \epsilon t, \quad \xi= \epsilon^{-1} x,
\end{equation}
and expand $u$ and $\rho$ in power series as
\begin{equation}
u = \epsilon^{2} (u_0+\epsilon u_1 + \cdots),  \quad
\rho = \epsilon (\rho_0+\epsilon \rho_1 + \cdots). \label{dep1}
\end{equation}
Then we have
${\partial_t} = \epsilon {\partial_\tau}$, ${\partial_x}= \epsilon^{-1} {\partial_\xi}$.
Substituting these relations  together with (\ref{dep1}) into Eqs. (\ref{2CHa})--(\ref{2CHc}), we obtain the following PDE for $u_0$ and
$\rho_0$ at the lowest order in $\epsilon$.
\begin{equation}
2\kappa u_{0,\xi}-u_{0,\tau \xi \xi}-2u_{0,\xi}u_{0,\xi \xi}-u_0u_{0,\xi
\xi \xi} + \sigma \rho_0 \rho_{0,\xi}=0,  \label{2CH-sworga}
\end{equation}
\begin{equation}
\rho_{0,\tau}+(\rho_0 u_0)_{\xi}=0.  \label{2CH-sworgb}
\end{equation}
Writing back above equations in terms of the original variables $t$, $x$, we
arrive at the two-component Hunter-Saxton equation (\ref{2CH-sweqa})--(\ref{2CH-sweqb}).

Both the two-component Camassa-Holm equation and two-component Hunter-Saxton equation are integrable in the sense that they admit the Lax pair
\begin{eqnarray}
\Psi _{xx} &=&\left( \frac{1}{4}-\lambda m+\lambda ^{2}\rho ^{2}\right) \Psi
,\quad \\
\Psi _{t} &=&-\left( \frac{1}{2\lambda }+u\right) \Psi _{x}+\frac{1}{2}%
u_{x}\Psi \,,
\end{eqnarray}
where $m=\kappa+ u -u_{xx}$ for the 2-CH equation (\ref{2CHa})--(\ref{2CHc}) and $m=\kappa -u_{xx}$ for the 2-HS equation (\ref{2CH-sweqa})--(\ref{2CH-sweqb}). The compatibility condition gives the 2-CH and 2-HS equations, respectively. Therefore, they both have infinite numbers of conservation laws and bi-Hamiltonian structures.


When $\kappa \ne 0$, the 2-HS equation (\ref{2CH-sweqa})--(\ref{2CH-sweqb}) is invariant under the following scaling transformations
\begin{equation}\label{Invariant}
   m \to \epsilon m , \quad u \to u/\epsilon , \quad  \rho \to  \epsilon \rho, \quad \kappa \to   \epsilon \kappa, \quad \sigma \to  \epsilon^2 \sigma, \quad \partial_x \to  \epsilon \partial_x\,.
\end{equation}
Therefore, without loss of generosity, we can fix one of the values for either $\kappa$ or $\sigma$.

In spite of many studies regarding the integrable 2-HS equation, as far as we are aware, its the two- and general $N$-soliton solution are not known yet. Therefore, the goal of the present paper is to find multi-soliton solution to 2-HS equation by using Hirota's bilinear method \cite{Hirotabook}.  The remainder of the present paper is organized as follows. In Section 2, through the hodograph transformation implied by the second equation in 2-HS equation, we propose a set of bilinear equations for 2-HS equation, and present its one- and two-soliton solutions by Hirota's method. In Section 3, Starting from bilinear equations of an extended two-dimensional Toda-lattice hierarchy, together with their Casorati determinant solution, we derive the general multi-soliton solution through a period 2-reduction. A detailed analysis for one- and two-soliton solutions shows that there are either smooth or loop soliton solution depending on the choices of parameters. Section 4 is devoted to some concluding remarks.
\section{Hodograph transformation and bilinear equations for two-component Hunter-Saxton equation}
The second equation (\ref{2CH-sweqb}) of 2-HS equation represents a conservation law, from which we can define a hodograph transformation, or sometimes called a reciprocal transformation
\begin{equation}\label{hodograph}
  dy=\rho dx - \rho u dt\,, \quad ds=dt\,.
\end{equation}
It then follows that
\begin{equation}\label{conversion_relation1}
\partial_x  = \rho \partial_y \,,  \quad \partial_t = \partial_s -\rho u  \partial_y\,,
\end{equation}
or \begin{equation}\label{conversion_relation2}
    \partial_y= \rho^{-1}\partial_x\,, \quad
\partial_s=\partial_t+u\partial_x\,.%
\end{equation}
Using (\ref{conversion_relation1}), Eq. (\ref{2CH-sweqb}) can be rewritten as
\begin{equation}\label{conv_lawysa}
\rho_s + \rho^2 u_y =0\,,
\end{equation}
or, equivalently,
\begin{equation}\label{conv_lawysb}
\left( \frac{1}{\rho} \right)_s = u_y \,.
\end{equation}
In view of the dependent variable transformations for the Camassa-Holm equation and short wave limit of the Camassa-Holm equation \cite{FMO,OMF_JPA}, the following transformation
\begin{equation}\label{trf1}
u=-(\ln g)_{ss}\,, \quad \rho=\frac{g^2}{fh}\,
\end{equation}
is suggested, where $f$, $g$ and $h$ are usually tau functions. Substituting  the transformation (\ref{trf1}) into (\ref{conv_lawysb}) and integrating with respect to $s$, we arrive at the following bilinear equation
\begin{equation}\label{bilinear1}
\left(\frac{1}{2}D_yD_s-1\right)g\cdot g=-fh\,,
\end{equation}
by taking an appropriate integration constant. Here $D$ is called Hirota $D$-operator defined by
\[
D_{s}^{n}D_{y}^{m}f\cdot g=\left( \frac{\partial }{\partial s}-\frac{%
\partial }{\partial s^{\prime }}\right) ^{n}\left( \frac{\partial }{\partial
y}-\frac{\partial }{\partial y^{\prime }}\right) ^{m}f(y,s)g(y^{\prime
},s^{\prime })|_{y=y^{\prime },s=s^{\prime }}\,.
\]%
We remark here that the bilinear equation (\ref{bilinear1}) can be viewed as the one for two-dimensional Toda-lattice, which will be made more
clear in the subsequent section. Note that the first equation of 2HS takes an alternative form
\begin{equation}
m_t + u m_x+2m u_x-\sigma \rho \rho_x=0,  \label{2HS3}
\end{equation}
\begin{equation}
m=2-u_{xx},  \label{2HS4}
\end{equation}
where $\kappa$ is set to be $2$ due to the invariant property (\ref{Invariant}) of the 2-HS equation (\ref{2CH-sweqa})--(\ref{2CH-sweqb}).  An substitution of (\ref{conversion_relation2}) into (\ref{2HS3}) and (\ref{2HS4})  yields
 \begin{equation}
m_s + 2m \rho u_y-\sigma \rho^2 \rho_y=0\,,  \label{2HSys1}
\end{equation}
 \begin{equation}
m=2+\rho(\ln \rho)_{ys}\,,  \label{2HSys2}
\end{equation}
respectively.
Eqs. (\ref{2HSys1})--(\ref{2HSys2}) imply the following two bilinear equations
\begin{eqnarray}
&& \left(D_y D_s+\frac{2}{\sqrt{\sigma}} D_s+ \sqrt{\sigma} D_y\right) f \cdot h=0,,  \\
&&\left(\frac{1}{\sqrt{\sigma}}D_s+1\right)f\cdot h=g^2 \,.
\end{eqnarray}
In summary, we have the following theorem.
\begin{theorem}
A set of bilinear equations
\begin{eqnarray}
&& \left(\frac{1}{2}D_yD_s-1\right)g\cdot g=-fh\,, \label{CHH1} \\
&& \left(\frac{1}{\sqrt{\sigma}}D_s+1\right)f\cdot h=g^2 \,,  \label{CHH2} \\
&& \left(D_y D_s+\frac{2}{\sqrt{\sigma}} D_s+ \sqrt{\sigma} D_y\right) f \cdot h=0\,, \label{CHH3}
\end{eqnarray}
give the two-component Hunter-Saxton equation (\ref{2CH-sweqa})--(\ref{2CH-sweqb}) with $\kappa=2$ through
the dependent transformation
\begin{equation}\label{trf-dependent}
   u = -(\ln g)_{ss}\,, \quad  \rho= \frac{g^2}{fh}\,,
\end{equation}
and the hodograph transformation
\begin{equation}\label{trf-hodograph}
   x=y-(\ln g)_{s}, \quad t=s\,.
\end{equation}
\end{theorem}
\begin{proof}
Dividing both sides of Eq.(\ref{CHH1}) by $g^2$, we have
\begin{equation}
(\ln g)_{ys}-1=-\frac{fh}{g^2}\,.  \label{CHH3n}
\end{equation}
Dividing both sides of Eqs. (\ref{CHH2}) and (\ref{CHH3}) by $fh$, we arrive at
\begin{equation}
\frac{1}{\sqrt{\sigma}}\left(\ln\frac{f}{h}\right)_s+1=\frac{g^2}{fh}\,,  \label{CHH2n}
\end{equation}
and
\begin{equation}
(\ln fh)_{ys} +\left(\left(\ln\frac{f}{h}\right)_s+\sqrt{\sigma}\right)\left(\left(\ln\frac{f}{h}\right)_y+\frac{2}{\sqrt{\sigma}}\right)-2=0\,,
\label{CHH1n}
\end{equation}
respectively.
Multiplying Eq.(\ref{CHH3n}) by 2, then subtracting it from Eq. (\ref{CHH1n}), one yields
\begin{equation}
\left(\ln\frac{fh}{g^2}\right)_{ys}+\left(\left(\ln\frac{f}{h}\right)_s+\sqrt{\sigma}\right)\left(\left(\ln\frac{f}{h}\right)_y+
\frac{2}{\sqrt{\sigma}}\right)
=2\frac{fh}{g^2}\,.  \label{CHH4n}
\end{equation}
By referring to Eq. (\ref{CHH2n}), it can be easily checked that
\[
\frac{\partial x}{\partial s} = -(\ln g)_{ss}= u, \quad \frac{\partial x}{\partial y} = 1-
(\ln g)_{ys}= \rho^{-1}\,
\]
which realizes the hodograph transformation (\ref{hodograph}) defined previously.
Note that Eq. (\ref{CHH3n}) can be written as
\begin{equation}  \label{CHH5n}
\left(\frac{1}{\rho} \right)_s = - (\ln g)_{yss} =  u_y\,,
\end{equation}
or equivalently
\begin{equation}  \label{CHH6n}
(\ln \rho)_s = -{\rho} u_y = - u_x \,.
\end{equation}
Combining Eq. (\ref{CHH4n}) with Eq. (\ref{CHH2n}) , we have
\begin{equation}  \label{CHH7n}
2+{\rho} \left( \ln \rho \right)_{ys} = \sqrt{\sigma} \rho^2 \left(\left(\ln\frac{f}{h}\right)_y+\frac{%
2}{\sqrt{\sigma}}\right) \,.
\end{equation}
Let us define
\begin{equation}\label{CHH8n}
m = 2+ {\rho} \left( \ln \rho \right)_{ys}\,,
\end{equation}
then it immediately follows
\begin{equation}\label{CHH9n}
m = 2-u_{xx}\,,
\end{equation}
by using Eq. (\ref{CHH6n}) and the conversion relation (\ref{conversion_relation2}).
Differentiating Eq. (\ref{CHH7n}) with respect to $s$,  we obtain
\begin{eqnarray}
m_s &=& 2\sqrt{\sigma} \rho \rho_s \left(\left(\ln\frac{f}{h}\right)_y+\frac{2}{\sqrt{\sigma}}\right) + \sqrt{\sigma} \rho^2
\left(\ln \frac{f}{h}\right)_{ys} \nonumber \\
 &=& 2\rho_s \frac{2+ {\rho} \left( \ln \rho \right)_{ys}}{\rho} + \sqrt{\sigma} \rho^2
 \left(\ln\frac{f}{h}\right)_{ys} \nonumber \\
&=& 2m (\ln \rho)_s + \sigma \rho^2 \rho_y \nonumber \\
&=& -2m u_x + \sigma  \rho \rho_x\,.  \label{CHH10n}
\end{eqnarray}
Here equations (\ref{CHH2n}), (\ref{CHH6n}) and (\ref{CHH7n}) are used. \\
Finally the hodograph transformation (\ref{conversion_relation2}) converts Eqs. (\ref{CHH6n}) and (\ref{CHH10n}) into
\begin{equation}  \label{2HS1}
(\partial_t+ u \partial_x) \rho =- \rho u_x \,,
\end{equation}
and
\begin{equation}  \label{HSw2}
\left( \partial_t + u \partial_x \right) m = -2m u_x + \sigma \rho \rho_x \,,
\end{equation}
respectively. The above two equations are nothing but the 2-HS equation with $m=2-u_{xx}$.
\end{proof}
\begin{remark}
Substituting (\ref{CHH2}) into (\ref{CHH2}), one obtains
\begin{equation}\label{CHH4}
    \left(D_y D_s+ \sqrt{\sigma} D_y-2 \right) f \cdot h=2g^2\,.
\end{equation}
If we assume $f = h$ as $\sigma \to 0$, then Eq. (\ref{CHH4}) converges to
\begin{equation}\label{CHH5}
    \left(\frac 12 D_y D_s-1 \right) f \cdot f=2g^2\,.
\end{equation}
It is worth pointing out that Eqs. (\ref{CHH3}) and (\ref{CHH4}) are actually the bilinear equations for period 2-reduction of two-dimensional Toda-lattice, which yield the sine-Gordon equation and further the short-pulse equation and the short wave limit of the CH equation.
 \end{remark}

In the last, we proceed to finding the one- and two-soliton solutions for 2-HS equation (\ref{2CH-sweqa})--(\ref{2CH-sweqb}) by Hirota's perturbation method \cite{Hirotabook}.
\par {\bf One-soliton solution:} To find one-soliton solution, we assume
\begin{equation}\label{1soliton}
  g=1+e^{ky+\omega s}\,, \quad f=1+ c e^{ky+\omega s}\,, \quad  h=1+ d e^{ky+\omega s}\,,
\end{equation}
 and substitute them into the bilinear equations (\ref{CHH3})--(\ref{CHH2}). As a result, we obtain the following algebraic relations
\begin{eqnarray}
 && \omega k  = 2-(c+d)\,, \quad  cd = 1\,,\nonumber \\
 && \frac{\omega}{\sqrt{\sigma}} (c-d)= 2-(c+d)\,,\nonumber \\
  && \omega k +\sqrt{\sigma} k (c-d) = 2(c-d)-4 \nonumber\,.
\end{eqnarray}
By choosing $c$ as a free parameter, we yield the one-soliton solution determined by
\begin{equation*}
  k  = \frac{1}{\sqrt{\sigma}}\frac{c^2-1}{c}\,, \quad \omega = \sqrt{\sigma} \frac{1-c}{1+c} \,, \quad d=\frac{1}{c}\,.
\end{equation*}
The tau functions found here leads to the one-soliton solution of the form
\begin{equation}
u=-\frac{\sigma}{4} \left( \frac{1-c}{1+c}\right)^2 \mbox{sech}^2\left( \frac{k}{2}(y-vs)\right) \,,
\label{2HS1solitona}
\end{equation}

\begin{equation}
\rho=\left( 1 +  \frac{(1-c)^2}{4c}  \mbox{sech}^2\left( \frac{k}{2}(y-vs)\right) \right)^{-1}\,,  \label{2HS1solitonb}
\end{equation}

\begin{equation}
x=y + \sqrt{\sigma} \frac{1-c}{1+c} \left( \frac{1}{1+e^{k(y-v s)}}-1 \right)\,, \quad t=s\,.
\label{2HS1solitonc}
\end{equation}
The amplitude of the soliton is $\sigma(1-c)^2/(4(1+c)^2)$, and the velocity in the $(y,s)$ plane is $v=\sigma c/(1+c)^2$.
From Eqs. (\ref{2HS1solitona})--\ref{2HS1solitona}), it is obvious that the conditions of $c>0$ and $c \ne 1$ must be imposed in order to assure a non-singular soliton solution. It is interesting to note that there is an upper bound for the amplitude of solitons, which is $\sigma/4$. When $c\to 0$ or $c \to \infty$, the amplitude of the soliton approaches this upper bound, meanwhile, the velocity approaches zero. When $c \to 1$, the amplitude approaches zero, whereas, the velocity approaches to the maximum of $\sigma/4$.
An example with $c=7$ and $\sigma=4$ is illustrated in Fig. 1.

\begin{figure}[htbp]
\centerline{
\includegraphics[scale=0.35]{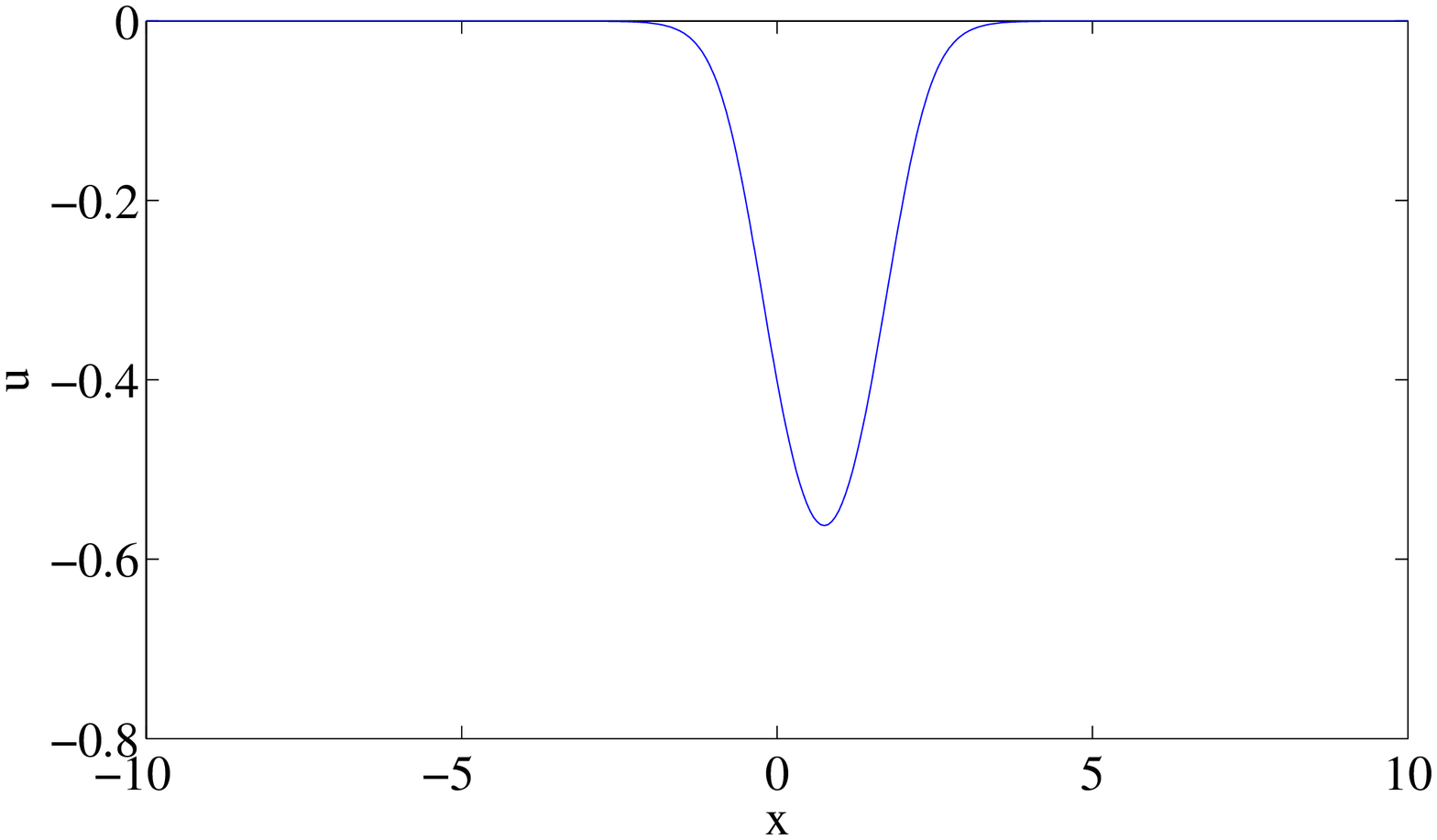}\quad
\includegraphics[scale=0.35]{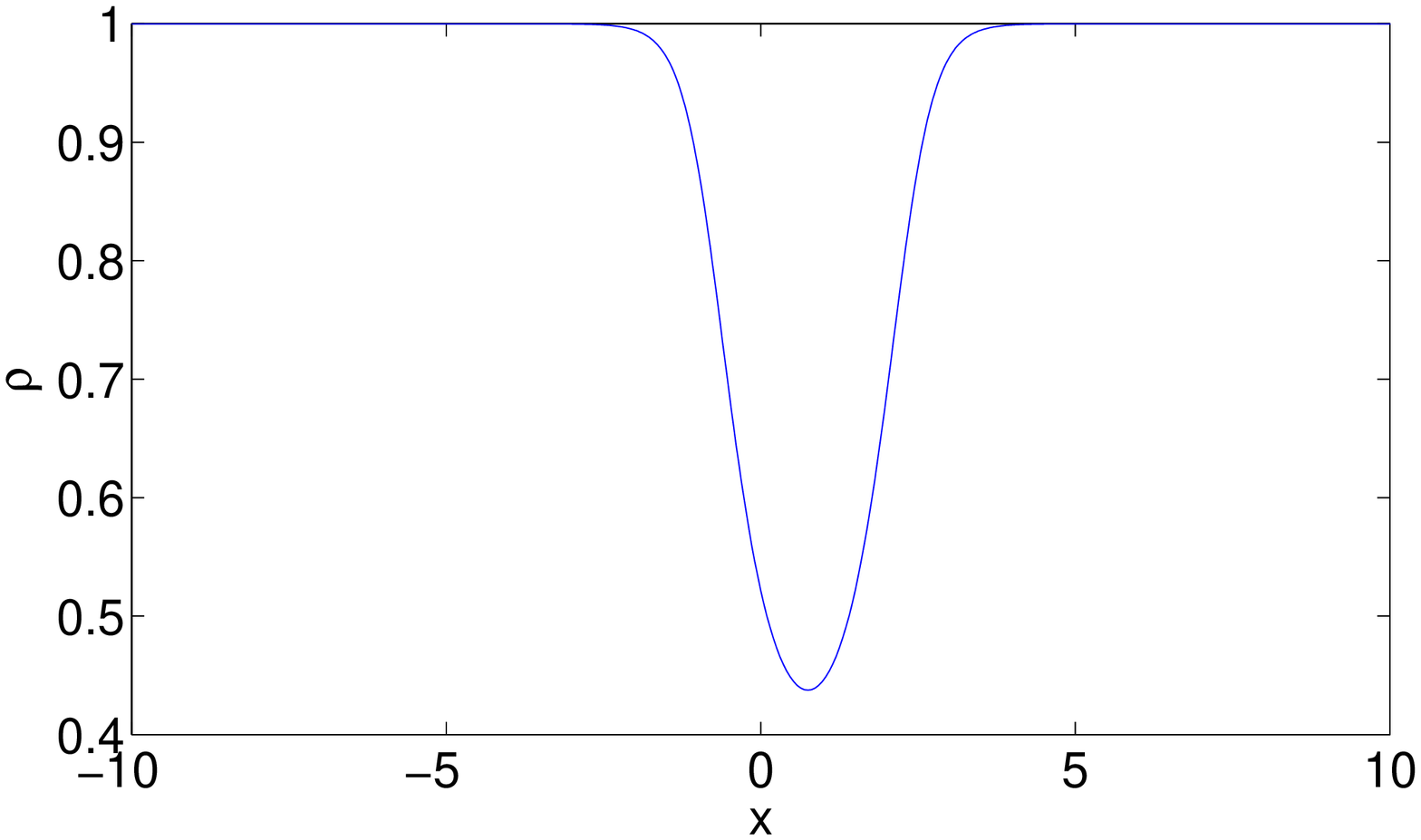}} \kern-0.25\textwidth
\hbox to
\textwidth{\hss(a)\kern17.em\hss(b)\kern5.3em} \kern+0.275\textwidth
\caption{An example of soliton solution for 2-HS equation: (a) $x$-$u$ profile; (b) $x$-$\rho$ profile.}
\label{fig1}
\end{figure}

\par {\bf Two-soliton solution:} Based on the obtained one-soliton solution, we assume the tau functions of two-soliton solution are of the form
  \begin{eqnarray}
 && g=1+e^{k_1y+\omega_1 s}+e^{k_2y+\omega_2 s}+ a_{12} e^{(k_1+k_2)y+(\omega_1+\omega_2) s}\,,\label{taug} \\
 && f=1+c_1e^{k_1y+\omega_1 s}+c_2 e^{k_2y+\omega_2 s}+ c_{12} e^{(k_1+k_2)y+(\omega_1+\omega_2) s}\,,\label{tauf} \\
  && h=1+c^{-1}_1e^{k_1y+\omega_1 s}+c^{-2}_2 e^{k_2y+\omega_2 s}+ d_{12} e^{(k_1+k_2)y+(\omega_1+\omega_2) s}\label{tauh}\,,
\end{eqnarray}
where
\begin{equation*}
  k_i  = \frac{c_i^2-1}{\sqrt{\sigma}c_i}\,, \quad \omega = \frac{1-c_i}{1+c_i}  \sqrt{\sigma}\,, \quad (i=1,2)\,.
\end{equation*}
After some tedious calculation, we obtain two-soliton solution determined by
\begin{equation*}
  a_{12}=\left( \frac{c_1-c_2}{1-c_1c_2}\right)^2\,, \quad c_{12}=c_1c_2 a_{12}\,, \quad d_{12}=c^{-1}_1c^{-1}_2 a_{12}\,.
\end{equation*}
We illustrate two-soliton solution interaction in Figs. 2-4. As shown in Fig.2, initially, we choose one soliton with $c=2.0$ located on the left, and one soliton with $c=7.0$ located on the right. They undertake a collision later on. The collision is elastic without any change of shape but a phase shift after the collision which is shown in Fig. 2-4 for $t=0,30,60$, respectively.

\begin{figure}[htbp]
\centerline{
\includegraphics[scale=0.35]{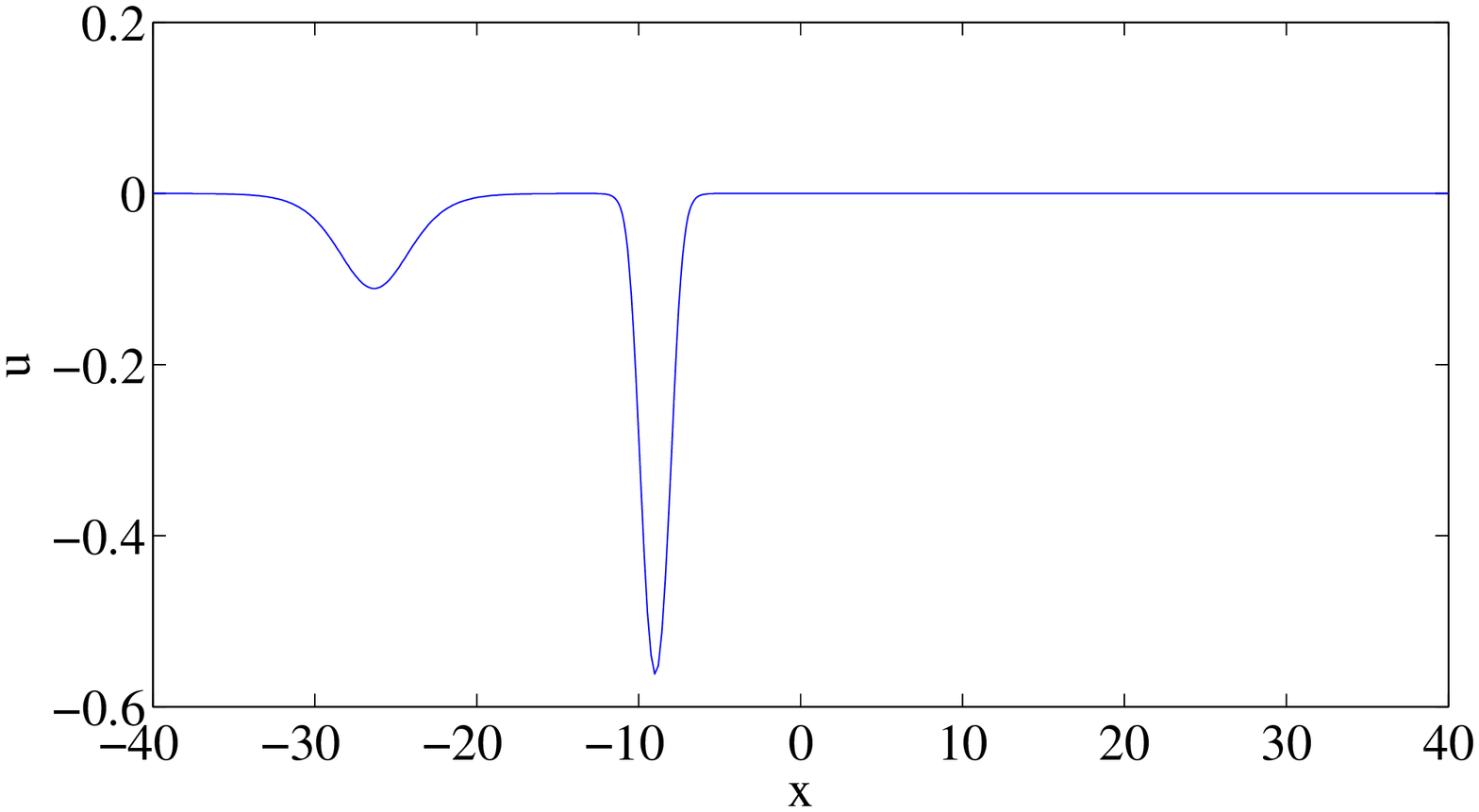}\quad
\includegraphics[scale=0.35]{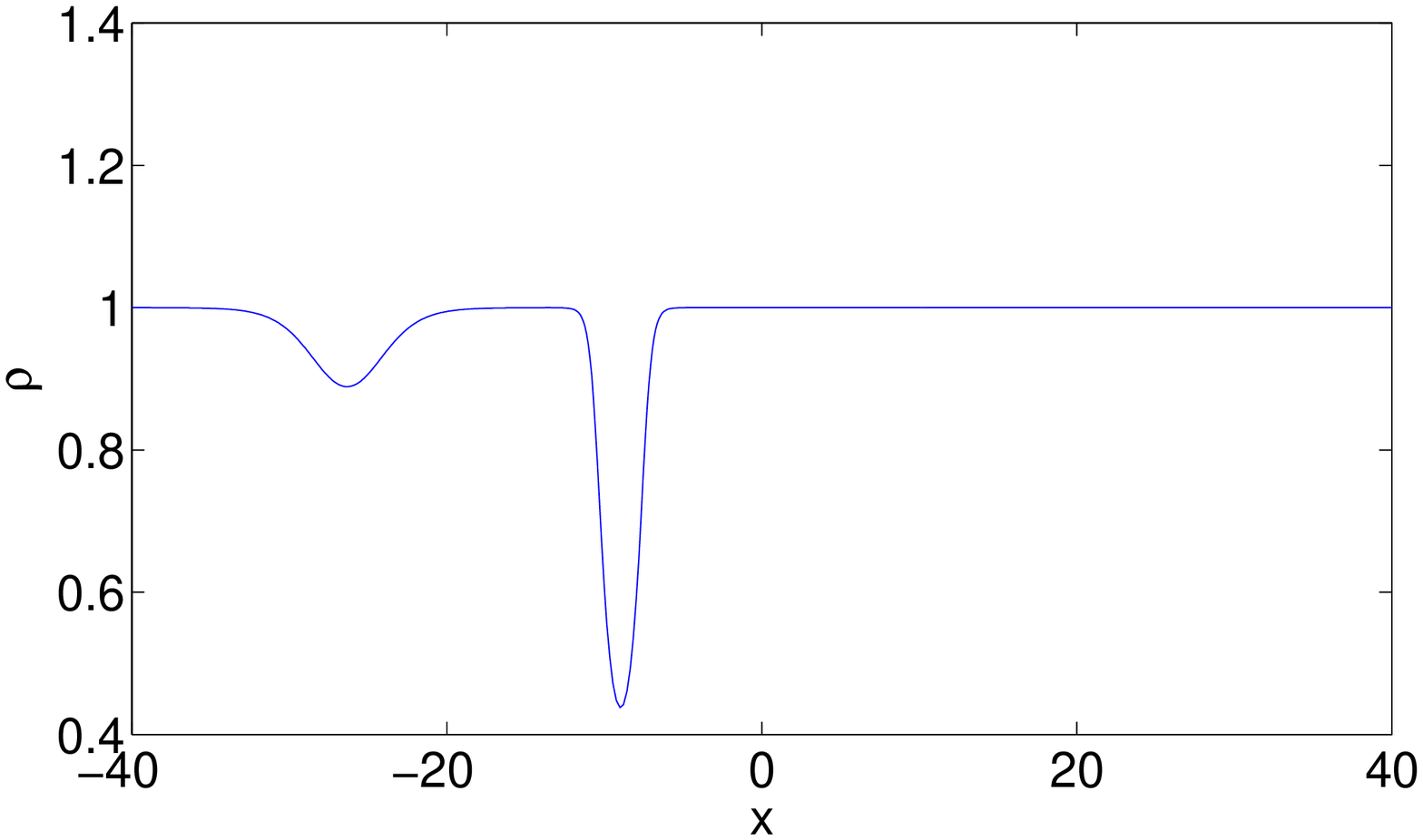}} \kern-0.25\textwidth
\hbox to
\textwidth{\hss(a)\kern17.em\hss(b)\kern5.3em} \kern+0.275\textwidth
\caption{Two-soliton interaction at $t=0$: $x$-$u$ profile; (b) $x$-$\rho$ profile.}
\label{fig2}
\end{figure}

\begin{figure}[htbp]
\centerline{
\includegraphics[scale=0.35]{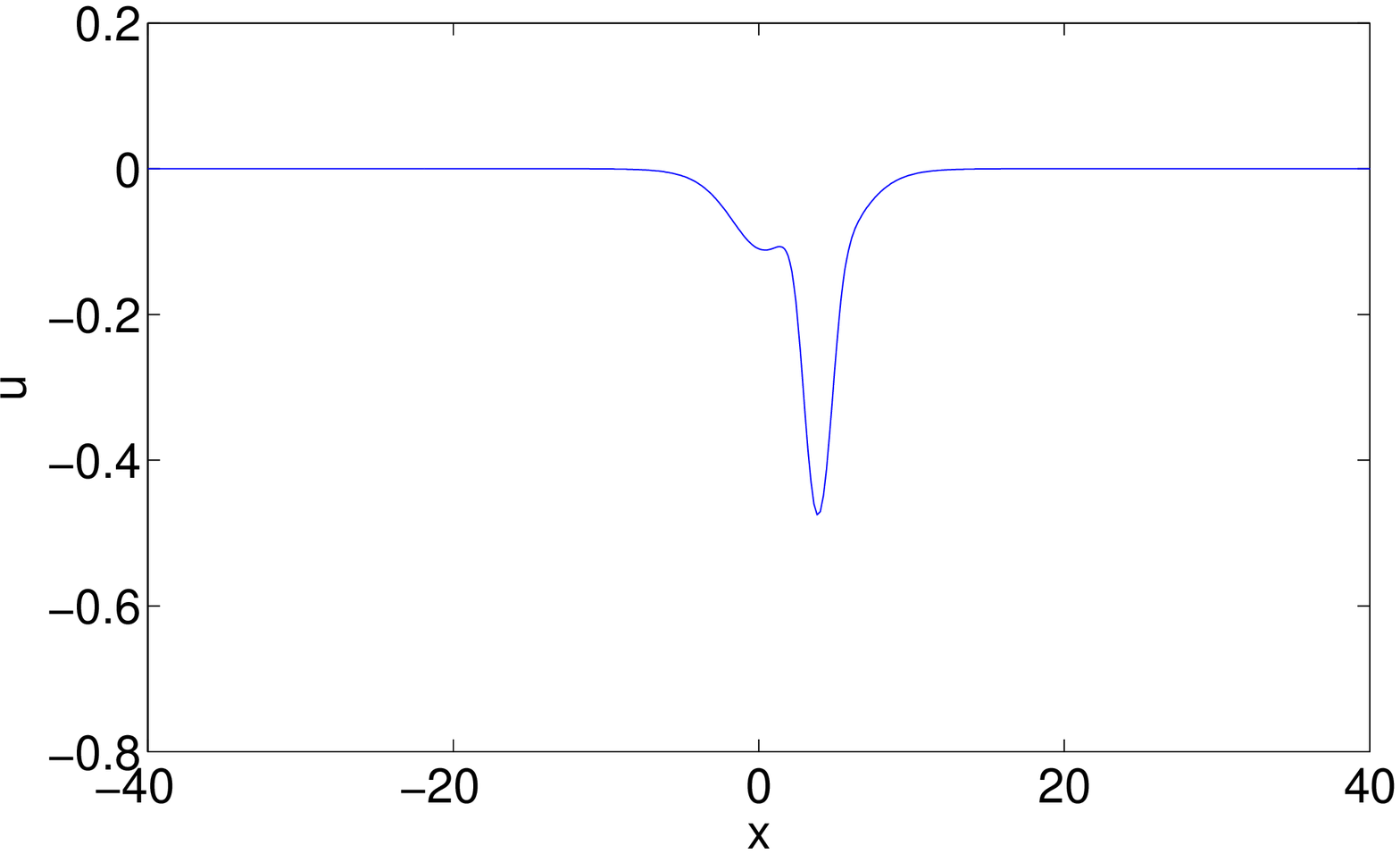}\quad
\includegraphics[scale=0.35]{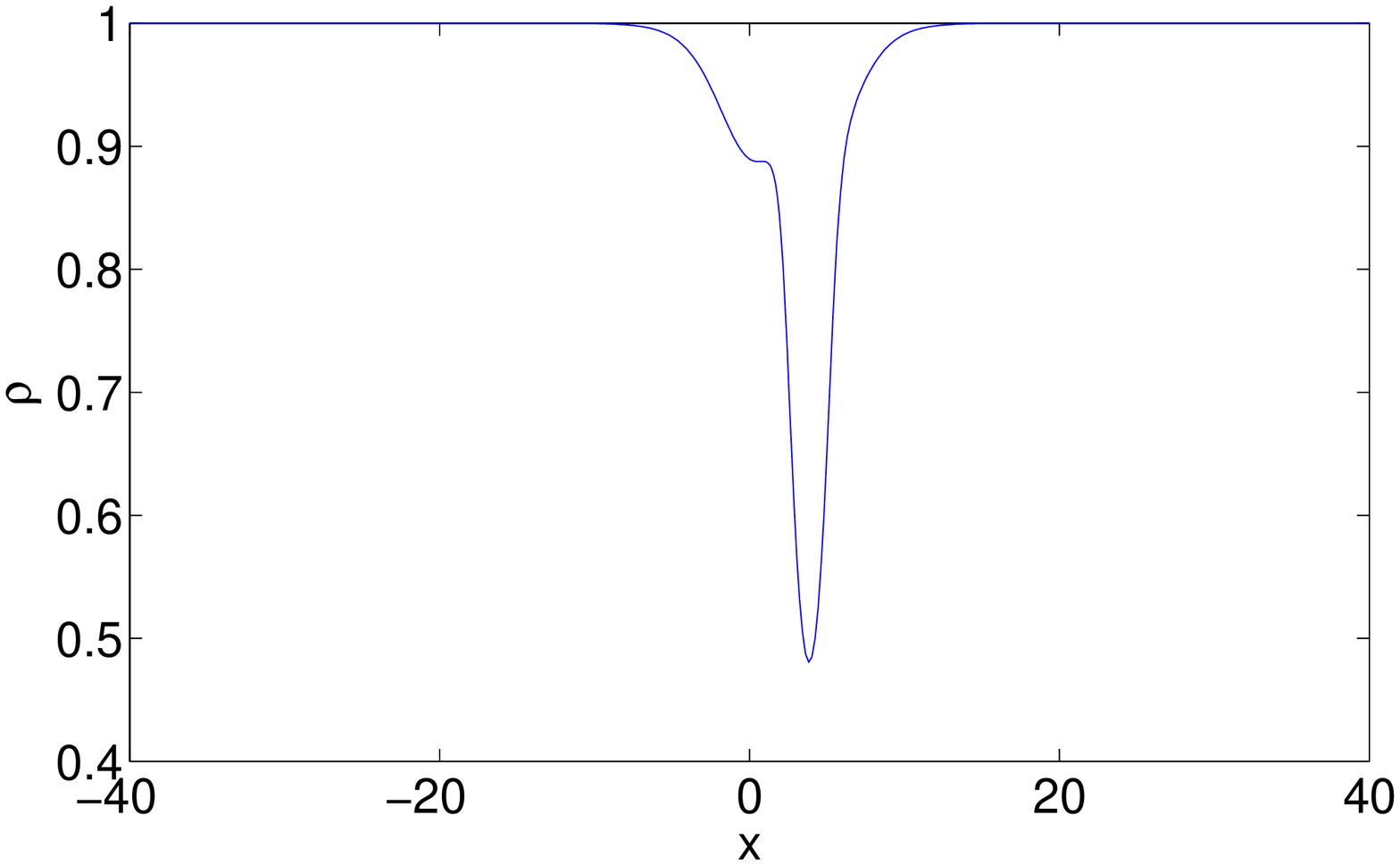}} \kern-0.25\textwidth
\hbox to
\textwidth{\hss(a)\kern17.em\hss(b)\kern5.3em} \kern+0.275\textwidth
\caption{Two-soliton interaction at $t=30$: $x$-$u$ profile; (b) $x$-$\rho$ profile.}
\label{fig3}
\end{figure}

\begin{figure}[htbp]
\centerline{
\includegraphics[scale=0.35]{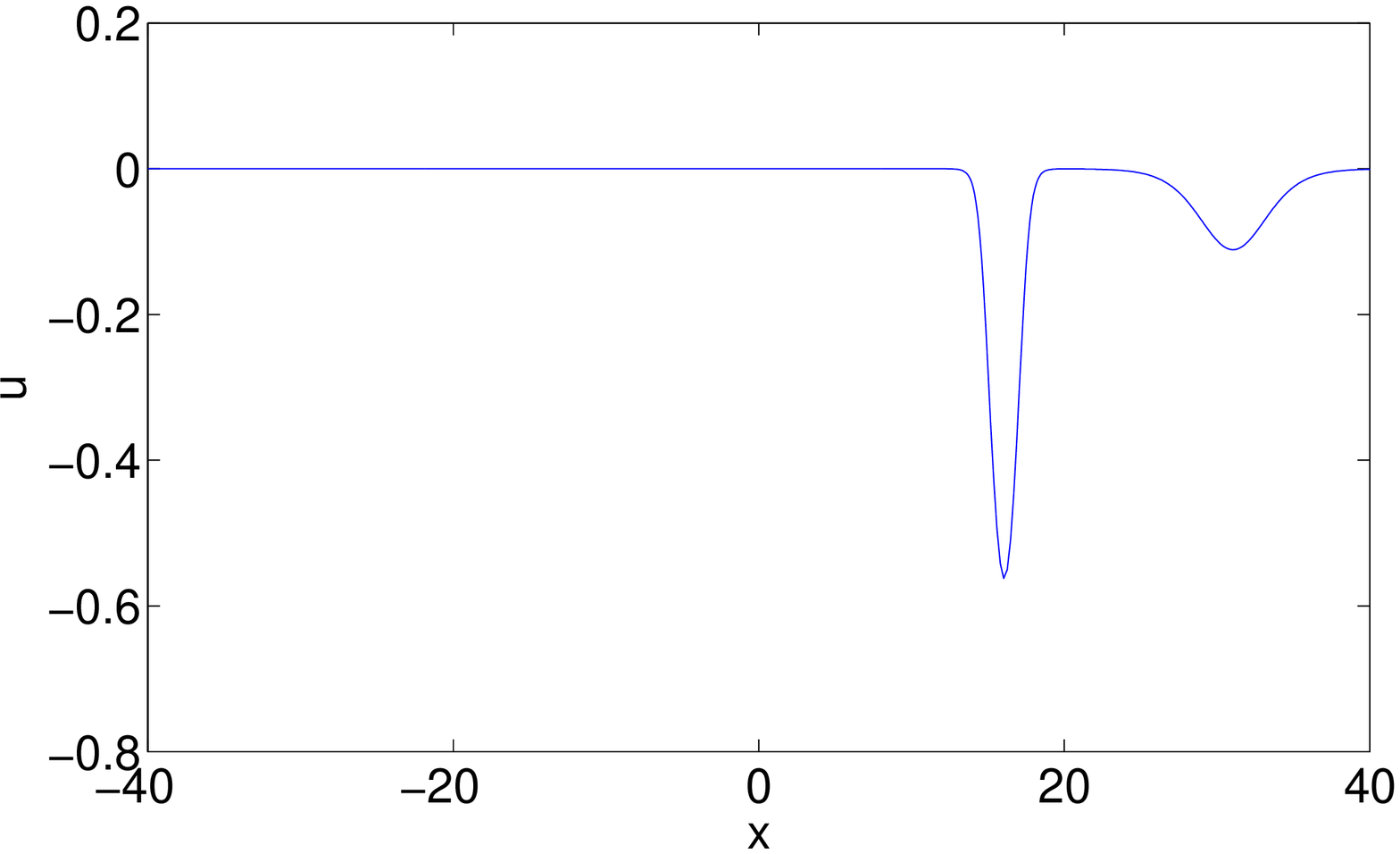}\quad
\includegraphics[scale=0.35]{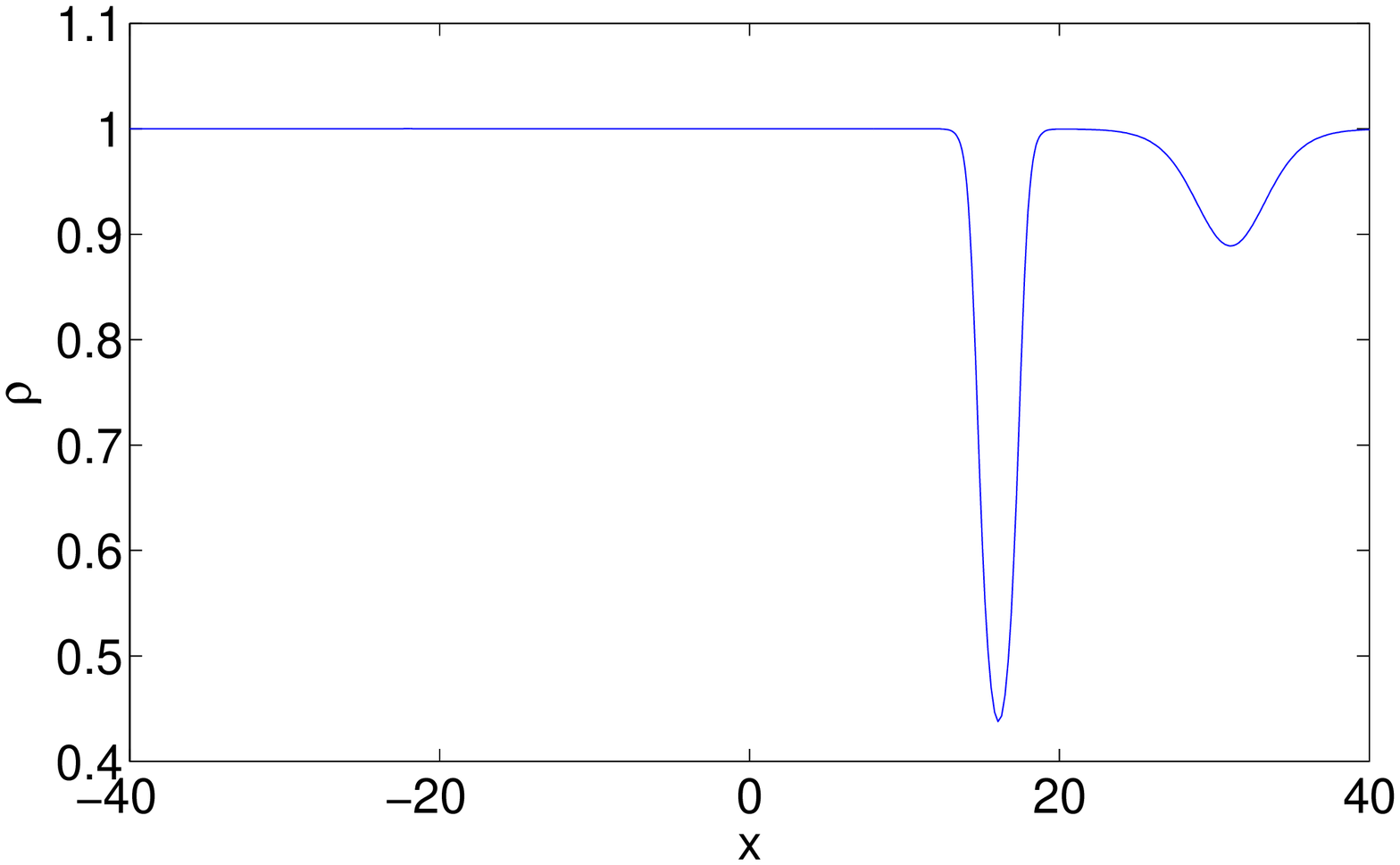}} \kern-0.25\textwidth
\hbox to
\textwidth{\hss(a)\kern17.em\hss(b)\kern5.3em} \kern+0.275\textwidth
\caption{Two-soliton interaction at $t=60$: $x$-$u$ profile; (b) $x$-$\rho$ profile.}
\label{fig4}
\end{figure}

\section{Multi-soliton solution to the two-component Hunter-Saxton equation}
We start with a Casorati determinant solution for the extended (deformed)  two-dimensional Toda hierarchy
\begin{equation}\label{tau}
  \tau_n=\left|\matrix{
\psi_1^{(n)} &\psi_1^{(n+1)} &\cdots &\psi_1^{(n+N-1)} \cr
\psi_2^{(n)} &\psi_2^{(n+1)} &\cdots &\psi_2^{(n+N-1)} \cr
\vdots       &\vdots         &       &\vdots           \cr
\psi_N^{(n)} &\psi_N^{(n+1)} &\cdots &\psi_N^{(n+N-1)}}\right|\,,
\end{equation}
where
\begin{equation}
\psi_i^{(n)}=a_{i,1}(p^{-1}_i-b)^ne^{\xi_i}+a_{i,2}(q^{-1}_i-b)^n e^{\eta_i}\,,
\end{equation} with
\begin{equation}
\xi_i=\frac{1}{p^{-1}_i-b}x_1+\frac{1}{(p^{-1}_i-b)^2}x_2+p^{-1}_i x_{-1} + p^{-2}_i x_{-2}+\xi_{i0}\,,
\end{equation}
\begin{equation} \label{taub}
\eta_i=\frac{1}{q^{-1}_i-b}x_1+\frac{1}{(q^{-1}_i-b)^2}x_2+q^{-1}_i x_{-1} + q^{-2}_i x_{-2}+\eta_{i0}\,.
\end{equation}
The following lemma gives bilinear equations satisfied by the above tau-functions presented.
\begin{lemma}
The tau functions $\tau_n$ given by (\ref{tau})--(\ref{taub}) satisfy the bilinear equations
\begin{equation}
\left( \frac{1}{2}D_{x_{1}}D_{x_{-1}}-1\right) \tau _{n}\cdot
\tau _{n}=-\tau _{n-1}\tau _{n+1}\,,
\label{2dbilinear1}
\end{equation}

\begin{equation}
\frac{1}{2}D_{x_{1}}(D_{x_{-2}}-2bD_{x_{-1}})\tau _{n}\cdot
\tau _{n}=-D_{x_{-1}}\tau _{n+1}\cdot \tau _{n-1},
\label{2dbilinear2}
\end{equation}

\begin{equation}
\frac{1}{2}D_{x_{2}}D_{x_{-1}}\tau _{n}\cdot \tau
_{n}=D_{x_{1}}\tau _{n+1}\cdot \tau _{n-1}\,,
\label{2dbilinear3}
\end{equation}

\begin{equation}
\left( \frac{1}{2}D_{x_{2}}(D_{x_{-2}}-2bD_{x_{-1}})-2\right)
\tau _{n}\cdot \tau _{n}=(D_{x_{1}}D_{x_{-1}}-2)\tau _{n+1}\cdot \tau
_{n-1}\,.
\label{2dbilinear4}
\end{equation}
\end{lemma}
\begin{proof}
The proof will be given by the technique developed by Hirota and Ohta \cite{Hirotabook,OhtaRIMS1989}.
Firstly, we can easily check that the following relations
\begin{eqnarray}
&&
\partial_{x_1}\psi_i^{(n)} = \psi_i^{(n-1)}\,,
\label{dispersion1}\\
&&
\partial_{x_2}\psi_i^{(n)} = \psi_i^{(n-2)}\,,
\label{dispersion2} \\
&&\partial_{x_{-1}} \psi_i^{(n)} = \psi_i^{(n+1)}+b\psi_i^{(n)}\,,
\label{dispersion3}\\
&&\partial_{x_{-2}} \psi_i^{(n)} =\psi_i^{(n+2)}+2b\psi_i^{(n+1)}+b^2\psi_i^{(n)}\,.
\label{dispersion4}
\end{eqnarray}
For simplicity, we introduce a convenient notation,
\begin{eqnarray}
 |{n_1},
  {n_2},
  \cdots,
  {n_N}|
 = \left|\matrix{
  \psi_1^{(n_1)} &\psi_1^{(n_2)}&\cdots
   &\psi_1^{(n_N)}\cr
  \psi_2^{(n_1)} &\psi_2^{(n_2)}&\cdots
   &\psi_2^{(n_N)} \cr
  \vdots                     &\vdots                     &
   &\vdots                     \cr
  \psi_N^{(n_1)} &\psi_N^{(n_2)} &\cdots
   &\psi_N^{(n_N)} \cr}
 \right|\,,
\end{eqnarray}
by which $\tau_n$ is rewritten as
\begin{eqnarray}
 \tau_n=|n, n+1,\cdots,n+N-1|\,.\label{n-tau}
\end{eqnarray}
Based on above relations (\ref{dispersion1})--(\ref{dispersion4}), we have the following relations regarding the derivatives of $\tau_n$
\begin{eqnarray}
&& \partial_{x_1}\tau_n=|n-1,n+1,\cdots,n+N-1|\,, \label{x1-dif}
\end{eqnarray}
\begin{eqnarray}
&& (\partial_{x_{-1}}-Nb)\tau_n=|n,n+1,\cdots,n+N-2,n+N|\,, \label{x-1-dif}
\end{eqnarray}
\begin{eqnarray}
&& (\partial_{x_{1}}(\partial_{x_{-1}}-Nb)-1)\tau_n=|n-1,n+1,\cdots,n+N-2,n+N|
\,. \label{xt-dif}
\end{eqnarray}
\begin{eqnarray}
&& (\partial_{x_{-2}}-2b\partial_{x_{-1}} +Nb^2)\tau_n=|n,n+1,\cdots,n+N-2,n+N+1|
\nonumber\\
&&  \quad -|n,n+1,\cdots,n+N-3,n+N-1,n+N|\,, \label{y-dif}
\end{eqnarray}
\begin{eqnarray}
&&  \partial_{x_{1}} (\partial_{x_{-2}}-2b\partial_{x_{-1}} +Nb^2)\tau_n
=|n-1,n+1,\cdots,n+N-2,n+N+1|\nonumber\\
&&  \quad -|n-1,n+1,\cdots,n+N-3,n+N-1,n+N|\,. \label{yt-dif}
\end{eqnarray}
\begin{eqnarray}
&& \partial_{x_{2}}(\partial_{x_{-1}}-Nb)\tau_n=|n-2,n+1,\cdots,n+N-2,n+N|\nonumber\\
&& \quad -|n-1,n,n+2,\cdots,n+N-2,n+N|\,, \label{xs-dif}
\end{eqnarray}
\begin{eqnarray}
&& \partial_{x_{2}}(\partial_{x_{-2}}-2b\partial_{x_{-1}} +Nb^2)\tau_n
=|n-2,n+1,\cdots,n+N-2,n+N+1|\nonumber\\
&& \quad -|n-1,n,n+2,\cdots,n+N-2,n+N+1|\nonumber\\
&& \quad -|n-2,n+1,\cdots,n+N-3,n+N-1,n+N|\nonumber\\
&& \quad +|n-1,n,n+2,\cdots,n+N-3,n+N-1,n+N|\,. \label{ys-dif}
\end{eqnarray}
Then the Pl\"ucker identity for determinants
\begin{eqnarray}
\fl &|n-1, n+1,\cdots,n+N-2,n+N|\times
     |n,n+1,\cdots,n+N-2,n+N-1| \cr
\fl  - &|n, n+1,\cdots,n+N-2,n+N|\times
     |n-1,n+1,\cdots,n+N-2,n+N-1| \cr
\fl  + &|n+1,\cdots,n+N-2,n+N-1,n+N|\times
     |n-1,n,n+1,\cdots,n+N-2|
= 0,&
\end{eqnarray}
gives
$$
(\partial_{x_1}(\partial_{x_{-1}}-Nb)-1)\tau_n\times\tau_n
-(\partial_{x_{-1}}-Nb)\tau_n\times\partial_{x_1}\tau_n+\tau_{n+1}\tau_{n-1}= 0\,,
$$
i.e.,
$$
\left(\partial_{x_1} \partial_{x_{-1}} \tau_n-\tau_n
\right) \tau_n
-\partial_{x_1} \tau_n \partial_{x_{-1}} \tau_n
 +\tau_{n+1}\tau_{n-1}
= 0\,,
$$
which is exactly the first bilinear equation (\ref{2dbilinear1}). Next, let us prove the second bilinear equation.
A subtraction of the following two Pl\"ucker determinant identities
\begin{eqnarray}
&& \fl \quad |n-1, n+1,\cdots,n+N-3,n+N-1,n+N|\times
     |n,n+1,\cdots,n+N-2,n+N-1| \nonumber\\
&& \fl  - |n,n+1,\cdots,n+N-3,n+N-1,n+N|\times
     |n-1,n+1,\cdots,n+N-2,n+N-1| \nonumber\\
&& \fl  + |n-1,n,n+1,\cdots,n+N-3,n+N-1|\times
     |n+1,n+2,\cdots,n+N-1,n+N|
= 0,\nonumber\\
\end{eqnarray}
\begin{eqnarray}
&& \fl \quad |n-1,n+1,\cdots,n+N-2,n+N+1|\times
     |n,n+1,\cdots,n+N-2,n+N-1| \nonumber\\
&& \fl  - |n,n+1,\cdots,n+N-2,n+N+1|\times
     |n-1,n+1,\cdots,n+N-2,n+N-1| \nonumber\\
&& \fl  + |n+1,\cdots,n+N-1,n+N+1|\times
     |n-1,n,n+1,\cdots,n+N-3,n+N-2|
= 0\,,\nonumber\\
\end{eqnarray}
leads to \begin{eqnarray*}
&& \partial_{x_1} \left(\partial_{x_{-2}}-2b\partial_{x_{-1}} +Nb^2\right)\tau_n\times\tau_n
-(\partial_{x_{-2}}-2b\partial_{x_{-1}} +Nb^2)\tau_n\times\partial_{x_1}\tau_n \\
&& \qquad +(\partial_{x_{-1}}-Nb)\tau_{n+1}\times\tau_{n-1}
-\tau_{n+1}(\partial_{x_{-1}}-Nb)\tau_{n-1}=0\,,
\end{eqnarray*}
i.e.,
$$
(\partial_{x_1} \partial_{x_{-2}} \tau_n) \tau_n-\partial_{x_{-2}} \tau_n \partial_{x_1}\tau_n
-2b((\partial_{x_1} \partial_{x_{-1}} \tau_n)\tau_n -\partial_{x_{-1}} \tau_n\partial_{x_{1}} \tau_n)
+(\partial_{x_{-1}}\tau_{n+1})\tau_{n-1}-\tau_{n+1}(\partial_{x_{-1}}\tau_{n-1})= 0\,,
$$
which is nothing but the second equation (\ref{2dbilinear2}).
On the other hand, the difference of the following two Pl\"ucker determinant identities
\begin{eqnarray}
&& \fl \quad |n-2, n+1,n+2,\cdots,n+N-2,n+N|\times
     |n,n+1,\cdots,n+N-2,n+N-1| \nonumber\\
&& \fl  - |n-2,n+1,n+2,\cdots,n+N-2,n+N-1|\times
     |n,n+1,n+2,\cdots,n+N-2,n+N| \nonumber\\
&& \fl  + |n-2,n,n+1,n+2,\cdots,n+N-2|\times
     |n+1,n+2,\cdots,n+N-2,n+N-1,n+N|
= 0,\nonumber\\
\end{eqnarray}
\begin{eqnarray}
&& \fl \quad |n-1, n,n+2,\cdots,n+N-2,n+N|\times
     |n,n+1,\cdots,n+N-2,n+N-1| \nonumber\\
&& \fl  - |n-1,n,n+2,\cdots,n+N-2,n+N-1|\times
     |n,n+1,n+2,\cdots,n+N-2,n+N| \nonumber\\
&& \fl  + |n-1,n,n+1,n+2,\cdots,n+N-2|\times
     |n,n+2,\cdots,n+N-2,n+N-1,n+N|
= 0\,,\nonumber\\
\end{eqnarray}
implies
$$
\partial_{x_2}(\partial_{x_{-1}}-Nb)\tau_n\times\tau_n
-\partial_{x_2} \tau_n\times(\partial_{x_{-1}}-Nb)\tau_n
+(\partial_{x_1}\tau_{n-1})\tau_{n+1}-\tau_{n-1}(\partial_{x_1 }\tau_{n+1})= 0\,,
$$
i.e.,
$$
(\partial_{x_2} \partial_{x_{-1}} \tau_n) \tau_n
-\partial_{x_2} \tau_n \partial_{x_{-1}} \tau_n
-(\partial_{x_{1}}\tau_{n+1})\tau_{n-1}+\tau_{n+1}(\partial_{x_{1}}\tau_{n-1})= 0\,.
$$
To prove the forth bilinear equation, we need the following four determinant identities
\begin{eqnarray}
&& \fl \quad |n-2,n+1,\cdots,n+N-2,n+N+1|\times
     |n,n+1,\cdots,n+N-2,n+N-1| \nonumber\\
&& \fl  - |n,n+1,\cdots,n+N-2,n+N+1|\times
     |n-2,n+1,\cdots,n+N-2,n+N-1| \nonumber\\
&& \fl  + |n-2,n,n+1,\cdots,n+N-2|\times
     |n+1,\cdots,n+N-2,n+N-1,n+N+1|
= 0,\nonumber\\
\end{eqnarray}
\begin{eqnarray}
&& \fl \quad |n-2,n+1,\cdots,n+N-3,n+N-1,n+N|\times
     |n,n+1,\cdots,n+N-1| \nonumber\\
&& \fl  - |n,n+1,\cdots,n+N-3,n+N-1,n+N|\times
     |n-2,n+1,\cdots,n+N-1| \nonumber\\
&& \fl  + |n-2,n,n+1,\cdots,n+N-3,n+N-1|\times
     |n+1,\cdots,n+N-1,n+N|
= 0,\nonumber\\
\end{eqnarray}
\begin{eqnarray}
&& \fl \quad |n-1,n,n+2,\cdots,n+N-2,n+N+1|\times
     |n,n+1,\cdots,n+N-1| \nonumber\\
&& \fl  - |n,n+1,\cdots,n+N-2,n+N+1|\times
     |n-1,n,n+2,\cdots,n+N-1| \nonumber\\
&& \fl  + |n-1,n,n+1,\cdots,n+N-2|\times
     |n,n+2,\cdots,n+N-1,n+N+1|
= 0,\nonumber\\
\end{eqnarray}
\begin{eqnarray}
&& \fl \quad |n-1,n,n+2,\cdots,n+N-3,n+N-1,n+N|\times
     |n,n+1,\cdots,n+N-1| \nonumber\\
&& \fl  - |n,n+1,\cdots,n+N-3,n+N-1,n+N|\times
     |n-1,n,n+2,\cdots,n+N-1| \nonumber\\
&& \fl  + |n-1,n,n+1,\cdots,n+N-3,n+N-1|\times
     |n,n+2,\cdots,n+N-1,n+N|
= 0.\nonumber\\
\end{eqnarray}
Taking an appropriate linear combination of these four bilinear identities, and substituting into the differential relations for tau functions
(\ref{x1-dif})--(\ref{ys-dif}), we arrive at
\begin{eqnarray*}
&&\fl (\partial_{x_2}(\partial_{x_{-2}}-2b\partial_{x_{-1}} +Nb^2)-2)\tau_n\times\tau_n
-(\partial_{x_{-2}}-2b\partial_{x_{-1}} +Nb^2)\tau_n\times\partial_{x_2}\tau_n
+\partial_{x_1}\tau_{n-1}\times(\partial_{x_{-1}}-Nb)\tau_{n+1} \\
&&\fl -(\partial_{x_1}(\partial_{x_{-1}}-Nb)-1)\tau_{n-1}\times\tau_{n+1}
-\tau_{n-1}(\partial_{x_1}(\partial_{x_{-1}}-Nb)-1)\tau_{n+1}
+(\partial_{x_{-1}}-Nb)\tau_{n-1}\times\partial_{x_1}\tau_{n+1}= 0\,,
\end{eqnarray*}
i.e.,
\begin{eqnarray*}
&&\fl (\partial_{x_2}\partial_{x_{-2}}\tau_n)\tau_n
-\partial_{x_{-2}}\tau_n\partial_{x_2}\tau_n
-2b((\partial_{x_2}\partial_{x_{-1}}\tau_n)\tau_n-\partial_{x_{-1}}\tau_n\partial_{x_2}\tau_n)
-2\tau_n\tau_n \\
&&\fl -(\partial_{x_1}\partial_{x_{-1}}\tau_{n+1})\tau_{n-1}
+\partial_{x_{-1}}\tau_{n+1}\partial_{x_1}\tau_{n-1}
+\partial_{x_1}\tau_{n+1}\partial_{x_{-1}}\tau_{n-1}
-\tau_{n+1}\partial_{x_1}\partial_{x_{-1}}\tau_{n-1}
+2\tau_{n+1}\tau_{n-1}
= 0,
\end{eqnarray*}
which is exactly the fourth equation (\ref{2dbilinear4}). The proof is complete.
\end{proof}
Next, we show the reduction process in order to construct multi-soliton solution to 2-HS equation. We start with a period 2 reduction by
requiring
$$
q_i=-p_i\,,
$$
which leads to
$\partial_{x_{-2}} \tau_n = c \tau_n$. Thus, the bilinear equations (\ref{2dbilinear2}) and (\ref{2dbilinear4})
are simplified into
\begin{equation}
b D_{x_1}D_{x_{-1}}\tau _{n}\cdot \tau _{n}= D_{x_{-1}}\tau _{n+1}\cdot \tau _{n-1},
\label{2dbilinear5}
\end{equation}
and
\begin{equation}
-\left(b D_{x_{2}}D_{x_{-1}}+2\right)
\tau _{n}\cdot \tau _{n}=(D_{x_{1}}D_{x_{-1}}-2)\tau _{n+1}\cdot \tau
_{n-1}\,,%
\label{2dbilinear6}
\end{equation}
respectively. A substitution of Eq. (\ref{2dbilinear1}) into Eq. (\ref{2dbilinear5}) leads to
\begin{equation}
\left( D_{x_{-1}} + 2b \right) \tau _{n+1}\cdot \tau _{n-1} =2 b \tau _{n}\cdot \tau _{n} \,,
\label{2dbilinear7}
\end{equation}
while a substitution of Eq.(\ref{2dbilinear3}) into Eq. (\ref{2dbilinear6}) yields
\begin{equation}
(D_{x_{1}}D_{x_{-1}}+2b D_{x_{1}} -2)\tau _{n+1}\cdot \tau
_{n-1} = -2 \tau _{n}\cdot \tau _{n} \,.%
\label{2dbilinear8}
\end{equation}
Finally, substituting Eq. (\ref{2dbilinear7}) into Eq. (\ref{2dbilinear8}), one obtains
\begin{equation}
(D_{x_{1}}D_{x_{-1}}+b^{-1} D_{x_{-1}} +2b D_{x_{1}})\tau _{n+1}\cdot \tau
_{n-1} = 0 \,.%
\label{2dbilinear9}
\end{equation}
Finally, if we assume $\tau_{0}=g$, $\tau_{1}=f$ and $\tau_{-1}=h$ and let $x_1=y$, $x_{-1}=s$ and $2b=\sqrt{\sigma}$, Eqs. (\ref{2dbilinear1}), (\ref{2dbilinear7}) and
(\ref{2dbilinear9}) are exactly a set of bilinear equations (\ref{CHH1})--(\ref{CHH2}), which derive 2-HS equation.
In summary, we can provide the $N$-soliton solution to 2-HS equation by the following theorem
\begin{theorem}
The 2-HS equation (\ref{2CH-sweqa})--(\ref{2CH-sweqb}) admits the $N$-soliton solution of the parametric form
\begin{equation}\label{multisoliton1}
  u=-(\ln g)_{ss}\,, \quad \rho=\frac{g^2}{fh}\,,
\end{equation}
\begin{equation}
   x=y-(\ln g)_{s}, \quad t=s\,,
\end{equation}
where $g=\tau_{0}$, $f=\tau_{1}$, $h=\tau_{-1}$ with the tau functions $\tau_n$ determined by
\begin{equation}
  \tau_n=\left|\matrix{
\psi_1^{(n)} &\psi_1^{(n+1)} &\cdots &\psi_1^{(n+N-1)} \cr
\psi_2^{(n)} &\psi_2^{(n+1)} &\cdots &\psi_2^{(n+N-1)} \cr
\vdots       &\vdots         &       &\vdots           \cr
\psi_N^{(n)} &\psi_N^{(n+1)} &\cdots &\psi_N^{(n+N-1)}}\right|\,,
\end{equation}
where
\begin{equation}
\psi_i^{(n)}=a_{i,1}\left(p^{-1}_i-\frac{\sqrt{\sigma}}{2}\right)^ne^{\xi_i}+a_{i,2}\left(-p^{-1}_i-\frac{\sqrt{\sigma}}{2}\right)^n e^{\eta_i}\,,
\end{equation} with
\begin{equation}
\xi_i=\frac{1}{p^{-1}_i-{\sqrt{\sigma}}/{2}}y+p^{-1}_i s +\xi_{i0}\,,
\end{equation}
\begin{equation}
\eta_i=-\frac{1}{p^{-1}_i+{\sqrt{\sigma}}/{2}}y-p^{-1}_i s +\eta_{i0}\,.
\end{equation}
Here $a_{i,1}$, $a_{i,2}$, $\xi_{i0}$ and $\eta_{i0}$ are arbitrary constants.
\end{theorem}
The proof is obvious based on the previous discussions. \\
To compare with the one- and two-soliton solution, we list the tau functions for one- and two-soliton solutions as follows.
\par {\bf One-soliton}: Let $a_{1,1}=a_{1,2}=1$, based on the above theorem, we have
\begin{equation}
  g=e^{\xi_1}+e^{\eta_1} \propto 1 + e^{\frac{2p^{-1}_1}{p^{-2}_1-\sigma^2/4}y + 2p^{-1}_1 s+\xi_{0}}\,,
\end{equation}
\begin{eqnarray}
 f &=& \left(p^{-1}_1-{\sqrt{\sigma}}/{2}\right) e^{\xi_1}- \left(p^{-1}_1+{\sqrt{\sigma}}/{2}\right) e^{\eta_1}\,, \nonumber\\
   &\propto&   1 -
  \frac{p^{-1}_1-{\sqrt{\sigma}}/{2}}{p^{-1}_1+{\sqrt{\sigma}}/{2}}  e^{\frac{2p^{-1}_1}{p^{-2}_1-\sigma^2/4} y+ 2p^{-1}_1 s+\xi_{0}}\,,
\end{eqnarray}
and
\begin{eqnarray}
h &=& \left(p^{-1}_1-{\sqrt{\sigma}}/{2}\right)^{-1} e^{\xi_1}- \left(p^{-1}_1+{\sqrt{\sigma}}/{2}\right)^{-1} e^{\eta_1}\,, \nonumber\\
   &\propto&   1 -
  \frac{p^{-1}_1+{\sqrt{\sigma}}/{2}}{p^{-1}_1-{\sqrt{\sigma}}/{2}}  e^{\frac{2p^{-1}_1}{p^{-2}_1-\sigma^2/4}y + 2p^{-1}_1 s+\xi_{0}}\,.
\end{eqnarray}
  \par {\bf Two-soliton solution:} Based on the general $N$-soliton solution, the tau functions for two-soliton solution can be expanded as
\begin{eqnarray}
g &=&\left\vert
\begin{array}{cc}
e^{\xi_1}+e^{\eta_1}  & \left(p^{-1}_1-{\sqrt{\sigma}}/{2}\right) e^{\xi_1}- \left(p^{-1}_1+{\sqrt{\sigma}}/{2}\right) e^{\eta_1} \nonumber \\
e^{\xi_2}+e^{\eta_2}  & \left(p^{-1}_2-{\sqrt{\sigma}}/{2}\right) e^{\xi_2}- \left(p^{-1}_2+{\sqrt{\sigma}}/{2}\right) e^{\eta_2} %
\end{array}%
\right\vert  \\
&=& (p^{-1}_1-p^{-1}_2) e^{\eta_1+\eta_2} -(p^{-1}_1+p^{-1}_2)e^{\xi_1+\eta_2} \nonumber \\
&& \quad +(p^{-1}_1+p^{-1}_2)e^{\xi_2+\eta_1} -
(p^{-1}_1-p^{-1}_2) e^{\xi_1+\xi_2}\, \nonumber \\
& \propto & 1 - \frac{p^{-1}_1+p^{-1}_2}{p^{-1}_1-p^{-1}_2} e^{\xi_1-\eta_1}
 + \frac{p^{-1}_1+p^{-1}_2}{p^{-1}_1-p^{-1}_2} e^{\xi_2-\eta_2} -
 e^{\xi_1+\xi_2-\eta_1-\eta_2}\,.
\end{eqnarray}
By defining
\begin{equation}\label{trfd}
  - \frac{p^{-1}_1+p^{-1}_2}{p^{-1}_1-p^{-1}_2} e^{\xi_1-\eta_1} = e^{\xi_1-\eta_1+\xi_{10}}\,, \quad
\frac{p^{-1}_1+p^{-1}_2}{p^{-1}_1-p^{-1}_2} e^{\xi_2-\eta_2} =  e^{\xi_2-\eta_2+\xi_{20}}\,,
\end{equation}
and
\begin{equation}\label{trfe}
 c_i=-  \frac{p^{-1}_i-{\sqrt{\sigma}}/{2}}{p^{-1}_i+{\sqrt{\sigma}}/{2}}\,,
  \quad k_i = \frac{2p^{-1}_i}{p^{-2}_i-\sigma^2/4}\,,
  \quad \omega_i=2p^{-1}_i\,,
\end{equation}
we finally can show that $g$ is exactly the expression (\ref{taug}) obtained in previous section.
Similarly, we have
\begin{eqnarray}
 f &=&\left\vert
\begin{array}{cc}
\left(p^{-1}_1-\frac{\sqrt{\sigma}}{2}\right) e^{\xi_1}-\left(p^{-1}_1+\frac{\sqrt{\sigma}}{2}\right)e^{\eta_1}  & \left(p^{-1}_1-\frac{\sqrt{\sigma}}{2}\right)^2 e^{\xi_1}+ \left(p^{-1}_1+\frac{\sqrt{\sigma}}{2}\right)^2 e^{\eta_1} \nonumber \\
\left(p^{-1}_2-\frac{\sqrt{\sigma}}{2}\right)e^{\xi_2}-\left(p^{-1}_2+\frac{\sqrt{\sigma}}{2}\right)e^{\eta_2}  & \left(p^{-1}_2-\frac{\sqrt{\sigma}}{2}\right)^2 e^{\xi_2}+ \left(p^{-1}_2+\frac{\sqrt{\sigma}}{2}\right)^2e^{\eta_2} %
\end{array}%
\right\vert  \nonumber \\
& \propto & 1 - \frac{p^{-1}_1+p^{-1}_2}{p^{-1}_1-p^{-1}_2} \frac{p^{-1}_1-{\sqrt{\sigma}}/{2}}{p^{-1}_1+{\sqrt{\sigma}}/{2}} e^{\xi_1-\eta_1}
 + \frac{p^{-1}_1+p^{-1}_2}{p^{-1}_1-p^{-1}_2} \frac{p^{-1}_2-{\sqrt{\sigma}}/{2}}{p^{-1}_2+{\sqrt{\sigma}}/{2}} e^{\xi_2-\eta_2} \nonumber \\
 && \quad - \frac{\left(p^{-1}_1-{\sqrt{\sigma}}/{2}\right) \left(p^{-1}_2-{\sqrt{\sigma}}/{2}\right) }{\left(p^{-1}_1+{\sqrt{\sigma}}/{2}\right)
 \left(p^{-1}_2+{\sqrt{\sigma}}/{2}\right)} e^{\xi_1+\xi_2-\eta_1-\eta_2}\,,
\end{eqnarray}
and
\begin{eqnarray}
 h &=&\left\vert
\begin{array}{cc}
\left(p^{-1}_1-\frac{\sqrt{\sigma}}{2}\right)^{-1} e^{\xi_1}-\left(p^{-1}_1+\frac{\sqrt{\sigma}}{2}\right)^{-1}e^{\eta_1}
& e^{\xi_1}+ e^{\eta_1}  \\
\left(p^{-1}_2-\frac{\sqrt{\sigma}}{2}\right)^{-1}e^{\xi_2}-\left(p^{-1}_2+\frac{\sqrt{\sigma}}{2}\right)^{-1}e^{\eta_2}  & e^{\xi_2}+ e^{\eta_2} %
\end{array}%
\right\vert  \nonumber \\
& \propto & 1 - \frac{p^{-1}_1+p^{-1}_2}{p^{-1}_1-p^{-1}_2} \frac{p^{-1}_1+{\sqrt{\sigma}}/{2}}{p^{-1}_1-{\sqrt{\sigma}}/{2}} e^{\xi_1-\eta_1}
 + \frac{p^{-1}_1+p^{-1}_2}{p^{-1}_1-p^{-1}_2} \frac{p^{-1}_2+{\sqrt{\sigma}}/{2}}{p^{-1}_2-{\sqrt{\sigma}}/{2}} e^{\xi_2-\eta_2} \nonumber \\
 && \quad - \frac{\left(p^{-1}_1+{\sqrt{\sigma}}/{2}\right) \left(p^{-1}_2+{\sqrt{\sigma}}/{2}\right) }{\left(p^{-1}_1-{\sqrt{\sigma}}/{2}\right)
 \left(p^{-1}_2-{\sqrt{\sigma}}/{2}\right)} e^{\xi_1+\xi_2-\eta_1-\eta_2}\,,
\end{eqnarray}
which completely recover the two-soliton solution (\ref{tauf})--(\ref{tauh}) through (\ref{trfd})--(\ref{trfe}). Before we end this section, three remarks are given.
\begin{remark}
If we redefine the parameters as follows  $$
c_i=-  \frac{p^{-1}_i-{\sqrt{\sigma}}/{2}}{p^{-1}_i+{\sqrt{\sigma}}/{2}}\,,
  \quad k_i = \frac{2p^{-1}_i}{p^{-2}_i-\sigma^2/4}\,,
  \quad \omega_i=2p^{-1}_i\,,  $$
  then we recover exactly the one- and two-soliton solutions for 2-HS equation found in previous section.
\end{remark}
\begin{remark}
As $\sigma \to 0$, the tau function $g$ goes to the one representing $N$-soliton solution of the Hunter-Saxton equation
\cite{FMO}. This finding is consistent with the fact that the bilinear equations for the 2-HS equation converge to the ones for the Hunter-Saxton equation.
In this limiting case of one-soliton, $c_1=c_2=-1.0$, thus, $f=h=1-e^{k_1y+\omega_1s}$, $g=1+e^{k_1y+\omega_1s}$. Under this case, $\rho$ becomes singular at the peak point and the soliton is actually the cuspon type as mentioned in \cite{FMO}.
\end{remark}
\begin{remark}
 We should point out here that both the 2-CH equation and the 2-HS equation become the CH equation and the HS equation, respectively, in the limiting case of $\sigma \to 0$ not $\rho \to 0$. It is interesting that the quantity $\rho$ still exists in the CH equation and the HS equation, it is merely decoupled with the dependent variable $u$.
\end{remark}
\section{Concluding remarks}
In the present paper, we proposed a set of bilinear equations for the two-component Hunter-Saxton equation through defining appropriate dependent variable transformations and hodograph transformation. Based on this set of bilinear equations, we construct one-, two-soliton solution to the 2-HS equation. It is interesting that the one-soliton solution to the 2-HS equation is either smooth or loop soliton, whereas, the HS equation only admits the cuspon soliton solution and the 2-CH equation admits either smooth or cuspon solution. By a period 2 reduction of an extended two-dimensional Toda-lattice hierarchy, we provide and prove the $N$-soliton solution in Casorati determinant form. The expansion for one- and two- soliton solution agrees with the results obtained via Hirota's perturbation method.

In a series of work by one of the authors, the integrable discretizations for a class of integrable PDEs with hodograph transformation such as the short pulse equation, the CH and its short wave limit, i.e., the HS equation, the reduced Ostrovsky equation and multi-component short pulse equation have been developed \cite{SPE_discrete1,OMF_JPA,FMO,FMO-VE-discrete,FMO-MCSP-discrete}, and some of them have been successfully used as an integrable self-adaptive moving mesh method for the numerical simulations of these PDEs \cite{dCHcom,FMO-PJMI}. It is a further topic for us to explore the integrable discreitzation of the 2-HS equation.
\section*{Acknowledment}
This work is partially supported by the National Natural Science Foundation
of China (Nos.11428102 and 111471004).

\end{document}